\documentclass{sig-alternate}

\usepackage{amsthm,amssymb,amsmath}
\usepackage{graphics}
\usepackage{tikz}
\usepackage[plainpages=false,pdfpagelabels,colorlinks=true,citecolor=blue,hypertexnames=false]{hyperref}
\usepackage{color}

\newtheorem{theorem}{Theorem}

\newtheorem{corollary}[theorem]{Corollary}
\newtheorem{lemma}[theorem]{Lemma}

\newtheorem{experiment}[theorem]{Experiment}

\def\qed{\quad\rule{1ex}{1ex}}

\def\ord{\operatorname{ord}}
\def\height{\operatorname{ht}}

\let\set\mathbb
\def\lc{\operatorname{lc}}

\def\sgn{\operatorname{sgn}}

\def\<#1>{\langle#1\rangle}

\conferenceinfo{ISSAC '14,}{July 23--25, 2014, Kobe, Japan.}
\CopyrightYear{2014}
\crdata{978-1-4503-2501-1/14/07}
\doi{http://dx.doi.org/10.1145/2608628.2608634}

\begin{document}

\allowdisplaybreaks

\title{Bounds for D-Finite Closure Properties}

\numberofauthors{1}

\author{%
 \mathstrut Manuel Kauers\titlenote{Supported by the Austrian Science Fund (FWF) grant Y464-N18.}\\[\smallskipamount]
  \affaddr{\mathstrut RISC / Joh. Kepler University}\\
  \affaddr{\mathstrut 4040 Linz, Austria}\\
  \affaddr{\mathstrut mkauers@risc.jku.at}
}

\maketitle
\begin{abstract}
  We provide bounds on the size of operators obtained by algorithms 
  for executing D-finite closure properties.
  For operators of small order, we give bounds on the degree and on
  the height (bit-size). 
  For higher order operators, we give degree bounds that are parameterized
  with respect to the order and reflect the phenomenon that higher 
  order operators may have lower degrees (order-degree curves).
\end{abstract}








\section{Introduction}

A common way of representing special functions in computer algebra systems is
via functional equations of which they are a solution, or equivalently, by
linear operators which map the function under consideration to zero. Functions
admitting such a representation are called \emph{$D$-finite.} Arithmetic on
D-finite functions translates into arithmetic of operators. For such
computations it is common that the output may be much larger than the input. But
how large? This is the question we wish to discuss in this paper.

Estimates on the output size are interesting because they enter in a crucial way
into the complexity analysis for the corresponding operations, and because
algorithms based on evaluation/interpolation depend on an a-priori knowledge of
the size of the result. Bounds on the bit size are also needed for the design of
``two-line algorithms'' in the sense of~\cite{yen96}. For these reasons, there has
been some activity concerning bounds in recent years, especially for estimating
the sizes of operators arising from creative
telescoping~\cite{mohammed05,bostan10b,chen12b,chen12c,kauers14d}, i.e.,
algorithms for definite summation and integration.

The focus in the present paper is on \emph{closure properties.} Closure
properties refer to the fact that when $f$ and $g$ are D-finite, then so are
$f+g$ and $fg$ and various other derived functions. We say that the class of
D-finite functions is closed under these operations. Algorithms for ``executing
closure properties'' belong to the standard repertoire of computer
algebra since the 1990s~\cite{salvy94,mallinger96}. Our goal is to estimate the
size of operators annihilating $f+g$ or $fg$ depending on assumptions on the
sizes of operators annihilating $f$ and~$g$.

It is easy to get good bounds on the \emph{order} of the output of closure
property algorithms. Such bounds are well
known~\cite{stanley80,stanley99,kauers10j}. We add here bounds on the
\emph{degree} of the polynomial coefficients of the output operators, and also
on their \emph{height,} which measures the size of the coefficients in the
polynomial coefficients. We also give degree bounds that are parameterized by
the order and reflect the phenomenon that the degree decreases as the order
grows. Although all these results are in principle obtained by the same
reasoning as the classical bounds on the order, actually computing them is
somewhat more laborious. We therefore believe that it is worthwhile working them
out once and for all and having them available in the literature for reference.

\subsection{Notation}

Let $R$ be an integral domain. 
We consider the Ore algebra $\set A=R[x][\partial]$ with the commutation rule
\[
  \partial p = \sigma(p)\partial + \delta(p)\quad(p\in R[x])
\]
where $\sigma\colon R[x]\to R[x]$ is a homomorphism and $\delta\colon
R[x]\to R[x]$ is a $\sigma$-derivation. For definitions of these notions and
further basic facts about Ore algebras, see~\cite{bronstein96}.
Two important examples of Ore algebras are the algebra of linear differential
operators (where $\sigma=\mathrm{id}$ and $\delta=\frac{d}{dx}$) and the algebra
of linear recurrence operators (where $\sigma(x)=x+1$, $\sigma|_R=\mathrm{id}$
and $\delta=0$).

Elements of Ore algebras are called operators. We can let them act on
$R[x]$-modules $\mathcal{F}$ of ``functions'' in such a way that $p\cdot f=pf$
for all $p\in R[x]$ and $f\in\mathcal{F}$ and $(L+M)\cdot f=(L\cdot f)+(M\cdot
f)$ and $(LM)\cdot f=L\cdot(M\cdot f)$ for all $L,M\in\set A$ and all
$f\in\mathcal{F}$. A function $f\in\mathcal{F}$ is then called D-finite
(with respect to the action of $\set A$ on $\mathcal{F}$) if there exists
$L\in\set A\setminus\{0\}$ with $L\cdot f=0$.

Operators $L\in\set A$ have the form 
\[
  L=\ell_0+\ell_1\partial+\cdots+\ell_r\partial^r
\]
with $\ell_0,\dots,\ell_r\in R[x]$. When $\ell_r\neq0$, we call $\ord(P):=r$ the
\emph{order} of the operator~$L$. The \emph{degree} of $L$ is defined as the
maximum degree of its polynomial coefficients: $\deg(L):=\max_{i=0}^r\deg(\ell_i)$.

We assume that for the ground ring~$R$ a size function $\height\colon R\to\set R$ is
given with the properties $\height(0)=0$, $\height(a)\geq0$, $\height(a)=\height(-a)$, for all $a\in R$,
$\height(ab)\leq \height(a)+\height(b)$ for all $a,b\in R$, and 
\begin{equation}\label{eq:htsum}
 \height\Bigl(\sum_{i=1}^n a_i\Bigr)\leq\height(n-1)+\max_{i=1}^n \height(a_i)
\end{equation}
for any $a_1,\dots,a_n\in R$.  For example, when
$R=\set Z$, we can take $\height(a)=\log(1+|a|)$, and when $R=K[t]$, we can take
$\height(a)=1+\deg(a)$ (using $\deg(0):=-1$). The \emph{height} of a polynomial
$p=c_0+c_1x+\cdots+c_dx^d\in R[x]$ is defined as
$\height(p):=\max_{i=0}^d\height(c_i)$. Note that we have
\[
  \height(pq)\leq\height(\min\{\deg(p),\deg(q)\})+\height(p)+\height(q)
\]
for all $p,q\in R[x]$ (but of course $\height(1p)=\height(p)$). 
Observe that the height of a polynomial depends on the
basis of~$R[x]$ and that we use the standard basis $1,x,x^2,\dots$ in our
definition. The \emph{height} of an operator
$L=\ell_0+\ell_1\partial+\cdots+\ell_r\partial^r$ is defined as
$\height(L):=\max_{i=0}^r\height(\ell_i)$.

We also need to know how $\sigma$ and $\delta$ change the degree and the height
of elements of~$R[x]$. In order to avoid unnecessary notational and
computational overhead, let us assume throughout that
$\deg(\sigma(p))\leq\deg(p)$ and $\deg(\delta(p))\leq\deg(p)$ for all $p\in
R[x]$. This covers most algebras arising in applications. For the height, we
assume that a function $c\colon\set R^2\to\set R$ is given such that for all
$p,q\in R[x]$ with $\deg(p),\deg(q)\leq d$ and $\height(p),\height(q)\leq h$ we
have $\height(\pm\sigma(p)+\delta(q))\leq c(d,h)$. Note that this definition
implies $\height(\partial L)\leq c(\deg(L),\height(L))$ for every $L\in
R[x][\partial]$.  We assume that $c$ is nonnegative, in both arguments
non-decreasing, and satisfies a triangle inequality with respect to the second
argument. For example, for the algebra of differential operators we can take
$c(d,h)=\height(1)+\height(d)+h$, and a possible choice for the algebra of
recurrence operators is $c(d,h)=d\height(2)+h$.


We will need to iterate the function $c$ in the second argument, and we will write
the composed functions using the following notation:
\[
  c^{(n)}(d,h):=c(d,c^{(n-1)}(d,h)),\quad
  c^{(0)}(d,h):=h
\]
We assume that this function is also non-decreasing with respect to~$n$. With
this notation we then have $\height(\partial^n L)\leq
c^{(n)}(\deg(L),\height(L))$, and more generally, using also height properties
stated earlier,
\begin{alignat}1
  \height(ML)&\leq\height(\ord(M))+\height(\min\{\deg(M),\deg(L)\})\notag\\
  &\qquad{}+\height(M)+c^{(\ord(M))}(\deg(L),\height(L))\label{eq:ML}
\end{alignat}
for any two operators $L,M\in R[x][\partial]$. It is also not difficult to see
that when $p\in R[x]$ and $n\in\set N$, then for $p^{[n]}:=p\sigma(p)\cdots\sigma^{n-1}(p)$
we have
\begin{equation}\label{eq:rising}
 \height(p^{[n]})
 \leq (n-1)\height(\deg(p))+n\,c^{(n-1)}\bigl(\deg(p),\height(p)\bigr). 
\end{equation}


\subsection{Argument Structure}

If the function $f_1$ is annihilated by an operator $L_1$ and the function $f_2$
is annihilated by another operator~$L_2$, and if $L$ is an operator such that
$L=M_1L_1=M_2L_2$ for two other operators $M_1,M_2$, then $L$ annihilates the
function $f_1+f_2$. It is easy to see that such an operator $L$ always
exists. For, suppose $L_1=\ell_{1,0}+\cdots+\ell_{1,r}\partial^r$ and
$L_2=\ell_{2,0}+\cdots+\ell_{2,s}\partial^s$ are given.  Make an ansatz
$M_1=m_{1,0}+\cdots+m_{1,s}\partial^s$, $M_2=m_{2,0}+\cdots+m_{2,r}\partial^r$
with undetermined coefficients $m_{i,j}$ for two left multipliers. Compute the
coefficients of the operator $M_1L_1-M_2L_2$. They will be linear combinations
of the undetermined $m_{i,j}$ with coefficients in~$R[x]$. Equating coefficients
of $\partial^k$ in $M_1L_1-M_2L_2$ to zero gives a linear system over $R[x]$
with $(s+1)+(r+1)$ variables but only $(s+r)+1$ equations. This system must have
a nontrivial solution.

All the following arguments will be based on this idea: make an ansatz with
undetermined coefficients, compare coefficients, observe that there are more
variables than equations, conclude that there must be a solution. The technical
difficulty consists in deriving reasonably good estimates for the degrees and
the heights of the entries in the linear system. We then use the following lemma
to turn them into estimates on the size of the solution vectors.

\begin{lemma}\label{lemma:la}
  Let $A=((a_{i,j}))\in R[x]^{n\times m}$ be a matrix with 
  $\deg(a_{i,j})\leq d$ and $\height(a_{i,j})\leq h$ for all $i,j$. 
  Assume that $n<m$ so that the matrix has a nontrivial nullspace. 
  Then there exists a vector $v=(v_1,\dots,v_m)\in\ker A\subseteq R[x]^m\setminus\{0\}$ with
  $\deg(v_i)\leq nd$ and $\height(v_i)\leq \height(n!)+(n-1)\height(d)+n h$
  for all $i=1,\dots,m$. 
\end{lemma}
\begin{proof}
  Let $k$ be the rank of $A$ when viewed as matrix over $\mathrm{Quot}(R[x])$.
  By choosing a maximal linearly independent set of rows from $A$, we may assume
  that $A\in R[x]^{k\times m}$. By permuting the columns if necessary, we may
  further assume that $A=(A_1,A_2)$ for some $A_1\in R[x]^{k\times k}$ and
  $A_2\in R[x]^{k\times(m-k)}$ with $\det(A_1)\neq0$. By Cramer's rule, the
  vector $(v_1,\dots,v_m)$ with $v_{k+1}=-\det(A_1)$, $v_i=0$ ($i=k+2,\dots,m$),
  and $v_i=\det(A_{1|i})$ ($i=1,\dots,k$) where $A_{1|i}$ is the matrix obtained
  from $A_1$ by replacing the $i$th column by the first column of~$A_2$ belongs 
  to~$\ker A$. From the determinant formula
  \[
    \det(A_1)=\sum_{\pi\in S_k} \sgn(\pi) \prod_{i=1}^k a_{i,\pi(i)}
  \]
  it follows that $\deg(\det(A_1))\leq k d$ and 
  \[
    \height(\det(A_1))\leq \height(k!) + (k-1)\height(d) + k h.
  \]
  The same bounds apply for all the determinants $\det(A_{1|i})$ and hence for all
  coordinates~$v_i$ of the solution vector. Since $k\leq n$, the claim follows.
\end{proof}

\section{Common Left Multiples (``Plus'')}

For the differential case, the computation of common left multiples was studied
in detail by Bostan et al.\ for ISSAC 2012~\cite{bostan12b}. Their Theorem~6
says that if $L$ is the least common left multiple of differential operators
$L_1,\dots,L_n$, then $\ord(L)\leq r:=\sum_{k=1}^n\ord(L_k)$ and
\[
  \deg(L)\leq(n(r+1)-r)\max_{k=1}^n\deg(L_k).
\]
Without insisting in $\ord(L)$ being minimal, we reprove this result for
arbitrary Ore algebras and supplement it with a bound on the height
(Section~\ref{sec:2.1}). We then give a bound on the degree of common multiples
of non-minimal order and show that the degree decreases as the order
grows (Section~\ref{sec:2.2}).

\subsection{Operators of Small Order}\label{sec:2.1}

By a common left multiple of ``small order'', we mean a left multiple of $L_k$
whose order is at most the sum of the orders of the~$L_k$. The actual order of
the \emph{least} common left multiple may be smaller than this, for instance if 
some of the $L_k$ have a non-trivial common right divisor. For investigating the size
of common left multiples of small order, we compare coefficients of $\partial^i$
and consider linear systems with coefficients in~$R[x]$.

\begin{theorem}\label{thm:lclm}
  Let $L_1,\dots,L_n\in R[x][\partial]$, suppose $\deg(L_k)\leq d$
  and $\height(L_k)\leq h$ for $k=1,\dots,n$. 
  Then there is a common left multiple $L\in R[x][\partial]$ of $L_1,\dots,L_n$ with
  \begin{alignat*}1
    \ord(L)&\leq r := \sum_{k=1}^n\ord(L_k)\\
    \deg(L)&\leq (n(r+1)-r)d\\
    \height(L)&\leq \vphantom{\sum_{i=1}^n}
    \height(r)+\height((n(r+1)-r-1)!)\\[-8pt]
    &\quad{}+ (n(r+1)-r-1)\height(d)\\
    &\quad{}+ (n(r+1)-r) c^{(r)}(d,h)
  \end{alignat*}
\end{theorem}
\begin{proof}
  Make an ansatz for $n$ operators 
  $M_k=m_{k,0}+m_{k,1}\partial+\cdots+m_{k,r-\ord(L_k)}\partial^{r-\ord(L_k)}$
  with undetermined coefficients $m_{k,i}$ ($k=1,\dots,n$; $i=0,\dots,r-\ord(L_k)$).
  We wish to determine the $m_{k,i}\in R[x]$ such that 
  \[
    M_1L_1=M_2L_2=\cdots=M_nL_n\ (=L)
  \]
  by comparing coefficients with respect to~$\partial$ and solving the resulting
  linear system. Each $M_kL_k$ is an operator of order~$r$ whose coefficients
  are $R[x]$-linear combinations of the undetermined $m_{k,i}$ with
  coefficients that are bounded in degree by $d$ and in height
  by~$c^{(r)}(d,h)$. Coefficient comparison therefore leads to a system of
  linear equations with $\sum_{k=1}^n(r-\ord(L_k)+1)=n
  r-\sum_{k=1}^n\ord(L_k)+n=n(r+1)-r$ variables and $(n-1)(r+1)=n(r+1)-r-1$
  equations, which according to Lemma~\ref{lemma:la} has a solution vector with
  coordinates $v_i$ with $\deg(v_i)\leq (n(r+1)-r-1)d$ and $\height(v_i)\leq\height((n(r+1)-r-1)!)
  +(n(r+1)-r-2)\height(d)+(n(r+1)-r-1)c^{(r)}(d,h)$. If $M_1$ is an operator with
  coefficients of this shape, we get for $L=M_1L_1$ the size estimates as
  stated in the theorem by~\eqref{eq:ML}.
\end{proof}

Experiments indicate that the bounds on order and degree are tight for random operators. 
The bound on the height seems to be off by a constant factor. 

\begin{experiment}
  Consider the algebra $\set Z[x][\partial]$ with $\sigma(x)=x+1$ and
  $\delta=0$, set $\height(a)=\log(1+|a|)$ for $a\in\set Z$, and define
  $c(d,h)=d\height(2)+h$.  Instead of the recursive definition $c^{(r)}(d,h)$ we
  use $c^{(r)}(d,h)=d\height(r+1)+h$, which is justified because $\delta=0$ and
  $\height(\sigma^r(p))\leq\deg(p)\height(r+1)+\height(p)$ for every $p\in\set
  Z[x]$. 

  For two randomly chosen operators $L_1,L_2\in\set Z[x][\partial]$ of order,
  degree, and height~$s$ ($s=2,4,8,16,32$) we found that the order and degree of
  their least common left multiple match exactly the bounds stated in the
  theorem. The bound stated for the height seems to overshoot by a constant
  factor only. The data is given in the first two rows of the following
  table. In the third and fourth row we give the corresponding data for random
  operators in $R[x][\partial]$ with $R=\set Z_{1091}[t]$ and $\height(a)=\deg_t(a)$. In
  this case, we can take $c(d,h)=h$ and find that the bound of
  Theorem~\ref{thm:lclm}  is tight.

  \begin{center}
    \begin{tabular}{c|c|c|c|c|c}
      $s$ & 2 & 4 & 8 & 16 & 32 \\\hline
      height bound & 46.8 & 163.2 & 635.7 & 2646.3 & 11403.3 \\ 
      actual height & 17.3 &  76.7 & 347.6 & 1615.9 & 7575.4 \\\hline
      height bound & 12 & 40 & 144 & 544 & 2112 \\
      actual height & 12 & 40 & 144 & 544 & 2112 
    \end{tabular}
  \end{center}
\end{experiment}

\subsection{Order-Degree Curve}\label{sec:2.2}

The next result says that there exist higher order common left multiples of
lower degree.  Also this was already observed by Bostan et al.~\cite{bostan12b}, who in
their Section~6 show that the total arithmetic size (order times degree) of
higher order common multiples may be asymptotically smaller than the arithmetic
size of the least common left multiples. We state this result more explicitly as
a formula for an \emph{order-degree curve,} a hyperbola which constitutes a
degree bound $d$ in dependence of the order~$r$ of the multiple. More results 
on order-degree curves can be found in~\cite{chen12b,chen12c,jaroschek13a}. 

Technically, the result is again obtained by making an ansatz and comparing
coefficients, but this time, coefficients with respect to $x^j\partial^i$ are
compared, and the resulting linear system has coefficients in~$R$ rather than
in~$R[x]$. According to our experience, non-minimal order operators of low
degree have unreasonably large height, which is why in practice they are used
only in domains where the height is bounded, such as finite fields. We have 
therefore not derived height bounds for these operators. A result on the height
of non-minimal operators arising in creative telescoping can be found in~\cite{kauers14d}.

\begin{theorem}\label{thm:curve1}
  Let $L_1,\dots,L_n\in R[x][\partial]$ with $r_i=\ord(L_i)$ and $d_i=\deg(L_i)$
  for all~$i$. Let
  \[
    r\geq\sum_{k=1}^n r_k \text{ and } 
    d\geq \frac{(r+1)\sum_{k=1}^n d_k - \sum_{k=1}^n r_kd_k}{r+1 - \sum_{k=1}^n r_k}.
  \]
  Then there exists a common left multiple $L\neq0$ of
  $L_1,\dots,L_n$ with $\ord(L)\leq r$ and $\deg(L)\leq d$.   
\end{theorem}
\begin{proof}
  For $r,d\geq0$, make an ansatz for $n$ operators 
  \[
    M_k=\sum_{i=0}^{r-r_k}\sum_{j=0}^{d-d_k} m_{i,j,k}\,x^j\partial^i
  \]
  with undetermined coefficients $m_{i,j,k}$. We wish to determine the $m_{i,j,k}\in R$
  such that $M_1L_1=\cdots=M_nL_n$. 
  Then $M_kL_k$ is a common left multiple of $L_1,\dots,L_n$ of order at most~$r$ and degree at most~$d$.
  Coefficient comparison in the ansatz gives a linear system over $R$ with 
  \begin{alignat*}1
    &\sum_{k=1}^n\sum_{i=0}^{r-r_k}\sum_{j=0}^{d-d_k}1\\
    &=n (r{+}1)(d{+}1) - (r{+}1)\sum_{k=1}^n d_k
    - (d{+}1)\sum_{k=1}^n r_k + \sum_{k=1}^n r_kd_k
  \end{alignat*}
  variables and $(n-1)(r+1)(d+1)$ equations. It has a solution when
  \begin{alignat*}1
    (r{+}1)(d{+}1) - (r{+}1)\sum_{k=1}^n d_k 
    - (d{+}1)\sum_{k=1}^n r_k + \sum_{k=1}^n r_kd_k
    > 0.
  \end{alignat*}
  For $r$ and $d$ satisfying the constraints in the theorem, this inequality is true. 
\end{proof}

\begin{experiment}
  For three operators $L_1,L_2,L_3$ of order~$5$ and degree~$5$, the theorem
  says that they admit a common left multiple $L$ of order~$r$ and degree~$d$
  for every $r\geq15$ and $d\geq\frac{15(r-4)}{r-14}$. When we took three such
  operators at random from the algebra $\set Z[x][\partial]$ with
  $\sigma(x)=x+1$ and $\delta=0$, we found the degrees of their left multiple
  to match this bound exactly. We also found that the leading coefficient of
  their least common left multiple $L$ had removable factor of degree~$150$, so
  that the order-degree curve from Theorem~\ref{thm:curve1} matches the
  order-degree curve given in Theorem~9 in~\cite{jaroschek13a}.
\end{experiment}

\section{Polynomials (``Times'')}

If two functions $f_1$ and $f_2$ are annihilated by operators $L_1,L_2$,
respectively, then a common left multiple $L$ of $L_1,L_2$ annihilates the sum
$f_1+f_2$. We now turn to operators $L$ which annihilate the product~$f_1f_2$,
more generally, some function $f$ that depends polynomially on given functions
$f_1,\dots,f_n$ and their derivatives (or shifts). Before we can do this, we
need to specify how operators should act on products of functions.

\def\R{\mathbf{R}}

\subsection{Actions on Polynomial Rings}

Consider the ring extension 
\[
  \R=R[x][y_{i,j}:i=1,\dots,n,j\geq0].
\] 
We want the Ore algebra $R[x][\partial]$ to act on $\R$ in such a way that
$p\cdot P=pP$ and $\partial\cdot(pP)=\sigma(p)(\partial\cdot P)+\delta(p)P$ 
and $\partial\cdot(P+Q)=(\partial\cdot P)+(\partial\cdot Q)$
for all $p\in R[x],P,Q\in\R$, and $\partial\cdot y_{i,j}=y_{i,j+1}$ for all $i\in\set
N$. The polynomial variables $y_{i,j}$ are meant to represent the functions
$\partial^j\cdot f_i$. For the product, we require that there are
$\alpha,\beta,\gamma\in\{0,1,-1\}$ such that for all $P,Q\in\R$ we have
\begin{equation}\label{eq:product}
  \partial\cdot(PQ)=\alpha\,(\partial\cdot P)(\partial\cdot Q)
                   +\beta\,((\partial\cdot P)Q + P(\partial\cdot Q))
                   +\gamma\, PQ.
\end{equation}
To fix the action, it then remains to specify how $\partial$ acts on~$R[x]$.
Two canonical options are $\partial\cdot p=\sigma(p)$ and $\partial\cdot p=\delta(p)$.

In the first case, i.e., when ``$\partial=\sigma$'', we have
\begin{alignat*}1
  \sigma(p)&=\partial\cdot p=\partial\cdot(p1)=\sigma(p)(\partial\cdot1)+\delta(p)\\
  &=\sigma(p)\sigma(1)+\delta(p)
   =\sigma(p)+\delta(p),
\end{alignat*}
so this option is only available when $\delta=0$, and then, since
\begin{alignat*}1
  \partial\cdot(pq)&=\sigma(p)(\partial\cdot q)+0=(\partial\cdot p)(\partial\cdot q)
\end{alignat*}
for all $p,q\in R[x]\subseteq\R$ we must have $\alpha=1,\beta=\gamma=0$ for the
multiplication rule.

There is more diversity when ``$\partial=\delta$''. For example, in the
differential case ($\sigma=\mathrm{id}$, $\delta=\frac{d}{dx}$), we have
$\alpha=0,\beta=1,\gamma=0$, and for difference operators
($\delta=\Delta=\sigma-\mathrm{id}$) we have $\alpha=1,\beta=1,\gamma=0$.

Observe that the action of $R[x][\partial]$ on $\R$ is an extension of the
action of $R[x][\partial]$ on~$R[x]$. 

\def\Deg{\operatorname{Deg}}
\def\Ord{\operatorname{Ord}}

Every $P\in\R$ is a polynomial in the variables $y_{i,j}$ with coefficients that
are polynomials in $x$ over~$R$. We write $\height(P)$ for the maximum of the
heights of all the elements of~$R$ appearing in coefficients of the polynomial,
$\deg(P)$ for the degree of $P$ with respect to~$x$, and
$\Deg(P)=(D_1,\dots,D_n)$ where $D_i$ is the total degree of $P$ when viewed as
polynomial in the variables $y_{i,0},y_{i,1},y_{i,2},\dots$. For such degree
vectors, we write $(D_1,\dots,D_n)\leq (E_1,\dots,E_n)$ if $D_i\leq E_i$ for
all~$i$. Addition and maxima of such vectors is meant component-wise. 
We write $\Ord(P)=(S_1,\dots,S_n)$ if
$S_i\in\set N$ is the largest index such that the variable $y_{i,S_i}$ appears
in~$P$.

A polynomial $P$ with $\Deg(P)=(D_1,\dots,D_n)$ is called homogeneous if it is
homogeneous with respect to each group $y_{i,0},y_{i,1},\dots$ of variables,
i.e., if for every monomial $\prod_{i,j}y_{i,j}^{e_{i,j}}$ in $P$ and every 
$i=1,\dots,n$ we have $\sum_j e_{i,j}=D_i$. 

\def\tbinom#1#2{{\textstyle\binom{#1}{#2}}}

\begin{lemma}\label{lemma:3}
\begin{enumerate}
\item\label{lemma:3.1} For homogeneous polynomials $P,Q\in\R$ with
  $\Ord(P)=(S_1,\dots,S_n)$, $\Deg(P)=(D_1,\dots,D_n)$, 
  $\Ord(Q)=(T_1,\dots,T_n)$, $\Deg(Q)=(E_1,\dots,E_n)$, 
  we have
  \begin{alignat*}1
    \Ord(PQ)&\leq\max\bigl\{\Ord(P),\Ord(Q)\bigr\}\\
    \Deg(PQ)&\leq\Deg(P)+\Deg(Q)\\
    \deg(PQ)&\leq\deg(P)+\deg(Q)\\
    \height(PQ)&\leq\min\Bigl\{
       \sum_{i=1}^n\height(\tbinom{D_i+S_i}{D_i}),
       \sum_{i=1}^n\height(\tbinom{E_i+T_i}{E_i})\Bigr\}\\
       &\quad{} +\height(\min\{\deg(P),\deg(Q)\})\\
       &\quad{} +\height(P)+\height(Q)
  \end{alignat*}
  The first term in the expression for $\height(PQ)$ can be dropped if $P$ or $Q$ 
  have just one monomial, in particular, when $P$ or $Q$ are in~$R[x]$. 
\item\label{lemma:3.2} For $k\in\set N$ and a polynomial $P\in\R$ with $\Deg(P)=(D_1,\dots,D_n)\neq(0,\dots,0)$ we have
  \begin{alignat*}1
    \Ord(\partial^k \cdot P)&\leq\Ord(P)+(k,k,\dots,k)\\
    \Deg(\partial^k \cdot P)&\leq\Deg(P)\\
    \deg(\partial^k \cdot P)&\leq\deg(P)\\
    \height(\partial^k \cdot P)&\leq k\height(4)\sum_{i=1}^nD_i+c^{(k)}(\deg(P),\height(P))
  \end{alignat*}
\end{enumerate}
\end{lemma}
\begin{proof}
  1. The claims on orders and degrees are clear. For the claim on the height, observe
  that the coefficient of every monomial in $PQ$ is a sum over products~$pq$, where
  $p$ is a coefficient of $P$ and $q$ a coefficient of~$Q$. 
  We have 
  \begin{alignat*}1
    \height(pq)
    &\leq\height(\min\{\deg(p),\deg(q)\})+\height(p)+\height(q)\\
    &\leq\height(\min\{\deg(P),\deg(Q)\})+\height(P)+\height(Q). 
  \end{alignat*}
  When $p$ or $q$ have just one monomial, this completes the proof. 
  Otherwise, the number of summands $pq$
  in such a sum is bounded by the number of terms in $P$ and by the number
  of terms in~$Q$. The claim follows because a homogeneous polynomial of degree $D_i$
  in $S_i+1$ variables has at most $\binom{D_i+S_i+1-1}{D_i}$ terms. 

  2. It suffices to consider the case $k=1$. The general case follows by repeating
  the argument $k$ times. 
  The claims on orders and degrees follow directly from the product rule for the
  action of $\partial$ on $\R$ and the assumption that $\sigma$ and $\delta$ do
  not increase degree. 

  For the bound on the height, write $P=\sum_{\ell}p_\ell\tau_\ell$ for some
  $p_\ell\in R[x]$ and distinct monomials $\tau_\ell=\prod_{i,j}y_{i,j}^{e_{i,j}}$. 
  Then $\partial\cdot P=\sum_\ell(\sigma(p_\ell)(\partial\cdot\tau_\ell)+\delta(p_\ell)\tau_\ell)$ can be
  written as a sum $\sum_m q_m\sigma_m$ where the $\sigma_m$ are
  distinct monomials and the $q_m$ are sums of several polynomials
  $\sigma(p_\ell)$ or $-\sigma(p_\ell)$, and possibly one polynomial~$\delta(p_\ell)$. Each of
  these polynomials has height at most $c(\deg P,\height P)$. 
  We show that these sums have at most $4^D$ summands, where $D=D_1+\cdots+D_n$. Then the claim follows from
  \eqref{eq:htsum} and $\height(4^D)\leq\height(4)D$. 
  For one part, the number of summands is caused by the fact that for two fixed monomials
  $\sigma$ and~$\tau$, the application of $\partial$ to $\tau$ may create the monomial $\sigma$
  more than once. For the other part, a fixed term $\sigma$ may turn up for several terms~$\tau$.
  We need to discuss both effects. 

  For the first effect, for any two monomials $\sigma,\tau$ let $a_{\sigma,\tau}$ be the
  number of times the monomial $\sigma$ appears in $\partial\cdot\tau$, 
  and set $a_{\sigma,\tau}:=0$ if $\sigma$ or $\tau$ is not a monomial.
  We show by induction on $D$ that $a_{\sigma,\tau}\leq 2^D-1$. 
  For $D=1$ we have $\tau=y_{i,j}$ for some $i,j$, so $\partial\cdot\tau=y_{i,j+1}$,
  so $a_{\sigma,\tau}=[[\sigma=y_{i,j+1}]]\leq 1=2^1-1$, where $[[\cdot]]$ denotes
  the Iverson bracket. Now assume the bound is true for $D-1\geq1$. 
  Writing $\tau=\tilde\tau y_{i,j}$, the product rule~\eqref{eq:product} gives
  \[
    \partial\cdot(\tilde\tau y_{i,j})
    =\alpha (\partial\cdot\tilde\tau) y_{i,j+1}
    +\beta (\partial\cdot\tilde\tau) y_{i,j}
    +\beta \tilde\tau y_{i,j+1}
    +\gamma \tilde\tau y_{i,j}.
  \]
  It follows that
  \[
    a_{\sigma,\tau} \leq 
    \underbrace{
    \underbrace{\underbrace{a_{\sigma\!/y_{i,j+1},\tilde\tau}}_{\leq 2^{D-1}-1} + 
                         \underbrace{a_{\sigma\!/y_{i,j},\tilde\tau}}_{\leq 2^{D-1}-1}}_{\leq 2^D-2} + 
    \underbrace{[[\sigma{=}\tilde\tau y_{i,j}]] + [[\sigma{=}\tilde\tau y_{i,j+1}]]}_{\leq 1}
    }_{\leq 2^D-1},
  \]
  as claimed. 

  For the second effect, the total number of contributions to a coefficient $q_m$ in $\partial\cdot P$
  is bounded by $\sum_\tau a_{\sigma_m, \tau}\leq\sum_\tau (2^D-1)$. For the summation
  range, it suffices to let $\tau$ run over at most $2^D$ ``neighbouring'' terms of~$\sigma_m$, 
  for if $\sigma_m=y_{i_1,j_1}y_{i_2,j_2}\cdots y_{i_D,j_D}$, then the only terms $\tau$ for
  which $\partial\cdot\tau$ may involve $\sigma_m$ are those of the form
  \[
   y_{i_1,j_1-e_1}y_{i_2,j_2-e_2}\cdots y_{i_D,j_D-e_D}
  \]
  with $(e_1,\dots,e_D)\in\{0,1\}^D$. These are $2^D$ many.
\end{proof}


\subsection{Normal Forms}

If the functions $f_1,\dots,f_n\in\mathcal{F}$ are solutions of the operators $L_1,\dots,L_n$
then every function 
\[
  f=P(f_1,\dots,f_n,\ \dots\dots, \partial^m\cdot f_1,\dots,\partial^m\cdot f_n)
\]
where $P$ is a multivariate polynomial is again D-finite. To see this, it
suffices to observe that D-finiteness is preserved under addition,
multiplication, and application of~$\partial$, because the expression for $f$
can be decomposed into a finite number of these operations. For computing an
annihilating operator for~$f$, it suffices to have algorithms for performing
these closure properties and apply them repeatedly.  For obtaining a bound on the
order of an annihilating operator for~$f$, it suffices to have such bounds for
these operations.  However, it turns out that the bounds obtained in this way
overshoot significantly, and the corresponding algorithm has a horrible
performance.

It is much better to consider an algorithm that computes an annihilating
operator for $f$ directly from the polynomial~$P$, and this is what we will do
next. Observe that the relations $L_i\cdot f_i=0$ can be used to rewrite~$f$
as another polynomial $V$ in the functions $\partial^j\cdot f_i$ with
$j<\ord(L_i)$ only. In the following lemma, we analyze how the size of $V$ depends
on the size of~$P$.

\begin{lemma}\label{lemma:4}
  Let $L_1,\dots,L_n\in R[x][\partial]$, $r_i=\ord(L_i)$, $p_i=\lc(L_i)$ ($i=1,\dots,n$)
  and consider the ideal 
  \begin{alignat*}1
    \mathfrak{a}=\bigl\langle 
      &L_1\cdot y_{1,0}, \ \partial L_1\cdot y_{1,0}, \ \partial^2 L_1\cdot y_{1,0},\dots\\
      &L_2\cdot y_{2,0}, \ \partial L_2\cdot y_{2,0}, \ \partial^2 L_2\cdot y_{2,0},\dots\\
      &\dots\\
     &L_n\cdot y_{n,0}, \ \partial L_n\cdot y_{n,0}, \ \partial^2 L_n\cdot y_{n,0},\dots
    \bigr\rangle\subseteq\R.
  \end{alignat*}
  For every $m\in\set N$ and every homogeneous polynomial $P\in\R$ with $\Deg(P)=(D_1,\dots,D_n)$
  and $\Ord(P)<(r_1+m,\dots,r_n+m)$
  there exists a homogeneous polynomial $V\in\R$ with 
  \[
  \biggl(\prod_{i=1}^n (p_i^{D_i})^{[m]}\biggr) P\equiv V\bmod\mathfrak{a}
  \]
  and
  \begin{alignat*}1
    \Ord(V)&<(r_1,\dots,r_n)\\
    \Deg(V)&\leq(D_1,\dots,D_n)\\
    \deg(V)&\leq\deg(P) + m\sum_{i=1}^n D_i\deg(L_i)\\
    \height(V)&\leq \height(P)+
    m\sum_{i=1}^n\Bigl(\height(D_i+1)+ D_i\height(r_i+m)\\
    &{}+D_i\height(\deg(L_i))\vphantom{c^{(m)}}+D_i\,c^{(m)}(\deg(L_i),\height(L_i))\smash{\Bigr)}.
  \end{alignat*}
\end{lemma}
\begin{proof}
  Induction on~$m$. For $m=0$ there is nothing to show (take $V=P$). Suppose
  the claim is true for~$m-1$. Write
  \[
    P = \sum_{j_1=0}^{D_1}\cdots\sum_{j_n=0}^{D_n} P_{j_1,\dots,j_n} \prod_{i=1}^n y_{i,r_i+m-1}^{j_i}
  \]
  for some $P_{j_1,\dots,j_n}$ with $\Ord(P_{j_1,\dots,j_n})<(r_1+m-1,\dots,r_n+m-1)$
  and 
  $\Deg(P_{j_1,\dots,j_n})\leq (D_1-j_1,\dots,D_n-j_n)$. Then
  \begin{alignat*}1
    &\biggl(\prod_{i=1}^n \sigma^{m-1}(p_i^{D_i})\biggr)P\\
    ={}& \sum_{j_1=0}^{D_1}\cdots\sum_{j_n=0}^{D_n} \tilde P_{j_1,\dots,j_n} 
      \prod_{i=1}^n \bigl(\sigma^{m-1}(p_i) y_{i,r_i+m-1}\bigr)^{j_i}\\
    \equiv{}& \underbrace{\sum_{j_1=0}^{D_1}\cdots\sum_{j_n=0}^{D_n} \tilde P_{j_1,\dots,j_n} 
      \prod_{i=1}^n\tilde Q_i^{j_i}}_{=:\tilde P}\mod\mathfrak{a},
  \end{alignat*}
  where 
  \begin{alignat*}1
    \tilde P_{j_1,\dots,j_n}&=P_{j_1,\dots,j_n}\prod_{i=1}^n\sigma^{m-1}(p_i)^{D_i-j_i}\\
    \tilde Q_i&= \sigma^{m-1}(p_i) y_{i,r_i+m-1}
       - \bigl(\partial^{m-1}L_i\cdot y_{i,0}\bigr).
  \end{alignat*}
  First, because of $\Ord(\tilde P_{j_1,\dots,j_n}),\Ord(\tilde Q_i^{j_i})<(r_1+m-1,\dots,\break r_n+m-1)$
  we have $\Ord(\tilde P)<(r_1+m-1,\dots,r_n+m-1)$.
  Second, because of $\Deg(\tilde P_{j_1,\dots,j_n})=\Deg(P_{j_1,\dots,j_n})\leq (D_1-j_1,\dots,D_n-j_n)$
  and $\Deg(\prod_{i=1}^n\tilde Q_i^{j_i})\leq (j_1,\dots,j_n)$ we have
  $\Deg(\tilde P)\leq(D_1,\dots,D_n)$.
  Third, because of 
  \begin{alignat*}1
    \deg(\tilde P_{j_1,\dots,j_n})&\leq\deg(P)+\sum_{i=1}^n(D_i-j_i)\deg(L_i),\\
    \deg\biggl(\prod_{i=1}^n\tilde Q_i^{j_i}\biggr)&\leq \sum_{i=1}^n j_i\deg(L_i)
  \end{alignat*}
  we have $\deg(\tilde P)\leq\deg(P)+\sum_{i=1}^n D_i\deg(L_i)$. Fourth,
  because of these degree estimates and 
  \begin{alignat*}1
   \height(\tilde P_{j_1,\dots,j_n})&\leq\height(P)+\sum_{i=1}^n (D_i-j_i)\Bigl(\height(\deg(L_i))\\
   &\qquad{}+c^{(m)}(\deg(L_i),\height(L_i))\Bigr)
  \end{alignat*}
  and $\height(\tilde Q_i)\leq c^{(m)}(\deg(L_i),\height(L_i))$, we have, by $\sum_{i=1}^n j_i$ fold
  application of Lemma~\ref{lemma:3}.(\ref{lemma:3.1}),
  \begin{alignat*}1
    \height\Bigl(\tilde P_{j_1,\dots,j_n}\prod_{i=1}^n\tilde Q_i^{j_i}\Bigr)
    &\leq 
      \height(P) + \sum_{i=1}^n \Bigl(j_i\height(r_i+m)\\
    &\kern-5em{}+D_i\height(\deg(L_i))+D_ic^{(m)}(\deg(L_i),\height(L_i))\Bigr)
  \end{alignat*}
  and therefore, because $\tilde P$ is a sum of at most $\prod_{i=1}^n(D_i+1)$ such terms,
  \begin{alignat*}1
    \height(\tilde P)\leq{}&\height(P)+
    \sum_{i=1}^n\Bigl(\height(D_i+1)+ D_i\height(r_i+m)\\
    &{}+D_i\height(\deg(L_i))+D_i\,c^{(m)}(\deg(L_i),\height(L_i))\Bigr).
  \end{alignat*}
  By induction hypothesis, there exists $V$ such that
  \[
    \biggl(\prod_{i=1}^n (p_i^{D_i})^{[m-1]}\biggr)\tilde P\equiv V\bmod\mathfrak{a}
  \]
  with $\Ord(\tilde V)<(r_1,\dots,r_n)$, $\Deg(\tilde V)\leq(D_1,\dots,D_n)$,
  \begin{alignat*}1
    \deg(V)&\leq\deg(\tilde P)+(m-1)\sum_{i=1}^n D_i\deg(L_i)\\[-4pt]
                  &\leq\deg(P)+m\sum_{i=1}^n D_i\deg(L_i)\\
    \height(V)&\leq\height(\tilde P) + 
    (m{-}1)\sum_{i=1}^n\Bigl(\height(D_i{+}1)+ D_i\height(r_i{+}m{-}1)\\
    &{}+D_i\height(\deg(L_i))+D_i\,c^{(m-1)}(\deg(L_i),\height(L_i))\Bigr)\\
    &\leq\height(P)+
    m\sum_{i=1}^n\Bigl(\height(D_i+1)+ D_i\height(r_i+m)\\
    &{}+D_i\height(\deg(L_i))+D_i\,c^{(m)}(\deg(L_i),\height(L_i))\Bigr).
  \end{alignat*}
  Finally, because of 
  \begin{alignat*}1
    \biggl(\prod_{i=1}^n (p_i^{D_i})^{[m]}\biggr) P
    &= \biggl(\prod_{i=1}^n (p_i^{D_i})^{[m-1]}\biggr) 
      \biggl(\prod_{i=1}^n \sigma^{m-1}(p_i^{D_i})\biggr) P\\
    &\equiv 
      \biggl(\prod_{i=1}^n (p_i^{D_i})^{[m-1]}\biggr)\tilde P
    \equiv V\bmod\mathfrak{a},
  \end{alignat*}
  the polynomial $V$ has all the required properties.
\end{proof}

\subsection{Small Orders}

We are now ready to state the main result, which bounds the size of an operator
which annihilates a function given as a polynomial of $f_1,\dots,f_n$
and their derivatives or shifts.

We consider only homogeneous polynomials. If a function~$f$ is expressed in
terms of $f_1,\dots,f_n$ via an inhomogeneous polynomial~$P$, we can write
$P=P_1+P_2+\cdots+P_s$ where each $P_i$ is homogeneous, then apply the theorem
to the~$P_i$ separately, and then combine the resulting bounds using
Theorem~\ref{thm:lclm} to obtain a bound for~$P$. This is fair because it seems
that the overestimation explained at the beginning of the previous section only
happens when homogeneous components are not handled as a whole but subdivided
further into sums of even smaller polynomials.
  
\begin{theorem}\label{thm:main}
  Let $L_1,\dots,L_n\in R[x][\partial]$, $r_i=\ord(L_i)$ ($i=1,\dots,n$). 
  Let $\mathfrak{a}\subseteq\R$ be as in Lemma~\ref{lemma:4}.
  Let $P\in\R$ be a homogeneous polynomial with $\Deg(P)=(D_1,\dots,D_n)$
  and $\Ord(P)<(r_1,\dots,r_n)$. Then there exists an operator $L\in R[x][\partial]\setminus\{0\}$
  and a polynomial $p\in R[x]\setminus\{0\}$ with $pL\cdot P\in\mathfrak{a}$ and
  \begin{alignat*}1
    \ord(L)&\leq m:= \prod_{i=1}^n\tbinom{D_i+r_i-1}{D_i}\\[-4pt]
    \deg(L)&\leq m\deg(P)+m^2\sum_{i=1}^n D_i\deg(L_i)\\[-4pt]
    \height(L)&\leq \height(m!) +m\,c^{(m)}(\deg(P),\height(P)))\vphantom{\sum_{i=1}^n}\\[-4pt]
    &\quad{} +(m-1)\height\Bigl(\deg(P)+m\sum_{i=1}^n D_i\deg(L_i)\Bigr)\\[-4pt]
    &\quad{} + m^2\sum_{i=1}^n \Bigl(\height(4)D_i + \height(D_i+1) + D_i\height(r_i+m)\\[-3pt]
    &\qquad\qquad{} + \height(\deg(L_i)) +c^{(m)}(\deg(L_i),\height(L_i))\Bigr).
  \end{alignat*}
\end{theorem}
\begin{proof}
  Let $p_i=\lc(L_i)$ and $p=\prod_{i=1}^n(p_i^{D_i})^{[m]}$.
  We show that there exist $\ell_0,\dots,\ell_m\in R[x]$, not all zero, such that
  \begin{equation}\label{eq:2}
    p \sum_{k=0}^m \ell_k (\partial^k\cdot P)\in\mathfrak{a}.
  \end{equation}
  Consider the polynomials 
  \[
    P_k = \biggl(\prod_{i=1}^n (\sigma^k(p_i)^{D_i})^{[m-k]} \biggr) (\partial^k\cdot P)
  \]
  for $k=0,\dots,m$. 
  Bounds for $\partial^k\cdot P$ can be obtained from Lemma~\ref{lemma:3}.(\ref{lemma:3.2}).
  Applying Lemma~\ref{lemma:3}.(\ref{lemma:3.1}) with 
  $\prod_{i=1}^n \sigma^{k+j}(p_i)^{D_i}$ as $P$ and
  $\bigl(\prod_{i=1}^n (\sigma^k(p_i)^{D_i})^{[j]} \bigr) (\partial^k\cdot P)$ as~$Q$,
  for $j=0,\dots,m-k-1$ (so that there are altogether $m-k$ applications of the Lemma), we obtain
  \begin{alignat*}1
    \Ord(P_k)&<(r_1+k,\dots,r_n+k)\\
    \Deg(P_k)&=(D_1,\dots,D_n)\\
    \deg(P_k)&\leq\deg(P)+(m-k)\sum_{i=1}^n D_i\deg(L_i)\\
    \height(P_k)&\leq k\height(4)\sum_{i=1}^nD_i 
         + c^{(m)}(\deg(P),\height(P))\\
         &\kern-1em{}+
         (m-k)\sum_{i=1}^n \Bigl(\height(\deg(L_i)) 
         +c^{(m)}(\deg(L_i),\height(L_i))\Bigr)
  \end{alignat*}
  for all $k=0,\dots,m$, where we have used $c^{(k)}(\cdot,\cdot)\leq c^{(m)}(\cdot,\cdot)$,
  $\deg(p_i)\leq\deg(L_i)$, and $\height(p_i)\leq\height(L_i)$ to bring the expression for the
  height into the form stated here. 

  Using Lemma~\ref{lemma:4} and the above bounds for~$P_k$, we find for each $k\leq m$ a $V_k$ with
  \[
    \biggl(\prod_{i=1}^n (p_i^{D_i})^{[m]} \biggr) (\partial^k\cdot P)  
    =
    \biggl(\prod_{i=1}^n (p_i^{D_i})^{[k]} \biggr) P_k
    \equiv V_k\bmod\mathfrak{a}
  \]
  and $\Ord(V_k)<(r_1,\dots,r_n)$, $\Deg(V_k)=(D_1,\dots,D_n)$, 
  \begin{alignat*}1
    \deg(V_k)&\leq \deg(P)+m\sum_{i=1}^n D_i\deg(L_i)\\
    \height(V_k)&\leq k\height(4)\sum_{i=1}^nD_i 
         + c^{(m)}(\deg(P),\height(P))\\
         &\kern-1em{}+
         (m-k)\sum_{i=1}^n \Bigl(\height(\deg(L_i)) 
         +c^{(m)}(\deg(L_i),\height(L_i))\Bigr)\\
         &\kern-1em{}
         +k\sum_{i=1}^n\Bigl(\height(D_i+1)+D_i\height(r_i+k)
         +D_i\height(\deg(L_i))\\
         &\qquad{}+D_ic^{(k)}(\deg(L_i),\height(L_i))\Bigr)\\
         &\kern-1em\leq
         m\sum_{i=1}^n \Bigl(\height(4)D_i + \height(D_i+1) + D_i\height(r_i+m) \\ 
         &\qquad{}+ \height(\deg(L_i))+c^{(m)}(\deg(L_i),\height(L_i))\Bigr)\\
         &\kern-1em{}+c^{(m)}(\deg(P),\height(P)).\vphantom{\sum_{i=1}^n}
  \end{alignat*}
  In the ansatz $\sum_{k=0}^m\ell_k V_k\stackrel!=0$ with undetermined
  coefficients $\ell_0,\dots,\ell_m$, compare coefficients with respect to terms
  $\prod_{i,j} y_{i,j}^{e_{i,j}}$. This gives a linear system over $R[x]$ with
  $m+1$ variables, $\prod_{i=1}^n\binom{D_i+r_i-1}{D_i}=m$ equations, and with
  coefficients of degree at most $\deg(P)+m\sum_{i=1}^n D_i\deg(L_i)$
  and height at most 
  \begin{alignat*}1
         &m\sum_{i=1}^n \Bigl(\height(4)D_i+ \height(D_i+1) + D_i\height(r_i+m) + \height(\deg(L_i)) \\ 
         &\quad{}+c^{(m)}(\deg(L_i),\height(L_i))\Bigr)
         +c^{(m)}(\deg(P),\height(P)).
  \end{alignat*}
  By Lemma~\ref{lemma:la}, the theorem follows.  
\end{proof}

In its full generality, the theorem is a bit bulky. For convenient reference,
and as example applications, we rephrase it for three important special cases.
The first concerns simple products of the form $f_1f_2$ and powers~$f^k$, the
second is what is called ``D-finite Ore action'' in Koutschan's
package~\cite{koutschan10c}, and the third is the Wronskian.
Observe that the bound for the order of the symmetric power is lower than 
the bound that would follow by applying the bound for the symmetric product
$k-1$~times. 

\begin{corollary} \textbf{(Symmetric Product and Power)}
  \begin{enumerate}
  \item Let $L_1,L_2\in R[x][\partial]$ and let $f_1,f_2\in\mathcal{F}$ be solutions
    of $L_1,L_2$, respectively. Let $r_1=\ord(L_1)$ and $r_2=\ord(L_2)$ and let $d,h$
    be such that $\deg(L_1),\deg(L_2)\leq d$ and $\height(L_1),\height(L_2)\leq h$.
    Then there exists an operator $M\in R[x][\partial]$ with $M\cdot(f_1f_2)=0$ and
    \begin{alignat*}1
      \ord(M)&\leq r_1r_2,\qquad \deg(M)\leq 2d r_1^2r_2^2, \\
      \height(M)&\leq 
      \height((r_1r_2)!)+ (r_1r_2-1)\height(2r_1r_2 d)+ r_1r_2\height(1) \\
      &\kern-2em{}+ 2r_1^2r_2^2\bigl(2\height(4)+3\height(r_1r_2)
        +\height(d)+ c^{(r_1r_2)}(d,h)\bigr)
    \end{alignat*}
  \item Let $L\in R[x][\partial]$, $r=\ord(L)$, $d=\deg(L)$, $h=\height(L)$, and
    let $f\in\mathcal{F}$ be a solution of~$L$.  Let $k\in\set N$. Then there
    exists an operator $M\in R[x][\partial]$ with $M\cdot(f^k)=0$ and
    \begin{alignat*}1
      \ord(M)&\leq\tbinom{k+r}{k}=:m,\qquad\deg(M)\leq kdm^2,\\
      \height(M)&\leq\height(m!) + m\height(1) + (m-1)\height(mkd)\\
      &\quad{}+ m^2\bigl(k\height(4) + \height(k+1) + k\height(r+m)\\
      &\quad{}+\height(d)+ c^{(m)}(d,h)\bigr)
    \end{alignat*}
  \end{enumerate}
\end{corollary}
\begin{proof} For part~1, apply the theorem with $n=2$ and $P=y_{1,0}y_{2,0}$. 
  Note that $\Ord(P)=(0,0)<(r_1,r_2)$, $\Deg(P)=(1,1)$, $\deg(P)=0$, and $\height(P)=\height(1)$.
  For part~2, take $n=1$, $P=y_{1,0}^k$. Note that $\Ord(P)=0<r$, $\Deg(P)=k$, $\deg(P)=0$, and
  $\height(P)=\height(1)$. 
\end{proof}

\begin{corollary} \textbf{(Associates)}
  Let $L\in R[x][\partial]$ and let $f\in\mathcal{F}$ be a solution of~$L$.
  Let $A\in R[x][\partial]$ be another operator with $\ord(A)<\ord(L):=r$.
  Then $A\cdot f$ is annihilated by an operator $M$ with
  \begin{alignat*}1
    \ord(M)&\leq r,\qquad\deg(M)\leq r\deg(A) + r^2\deg(L),\\
    \height(M)&\leq \height(r!) + r\,c^{(r)}(\deg(A),\height(A))\\
    &\kern-2em{} + (r-1)\height(\deg(A) + r\deg(L))\\
    &\kern-2em{} + r^2\bigl(4\height(2) + \height(r) + \height(\deg(L)) + c^{(r)}(\deg(L),\height(L))\bigr)
  \end{alignat*}
\end{corollary}
\begin{proof}
  Apply Theorem~\ref{thm:main} with $n=1$ and $P=A\cdot y_{1,0}$.  Note that
  $\Ord(P)<r-1$, $\Deg(P)=D_1=1$, $\deg(P)=\deg(A)$, and
  $\height(P)=\height(A)$.  In the expression for the height, we used
  $\height(4)+\height(1+1)+\height(r+r-1)\leq3\height(2)+\height(2r)\leq
  4\height(2)+\height(r)$.
\end{proof}

\begin{corollary} \textbf{(Wronskian)}\label{corr:wr}
  Let $L_1,\dots,L_r\in R[x][\partial]$ be
  operators of order~$r$, degree~$d$ and height~$h$. Let
  $f_1,\dots,f_r\in\mathcal{F}$ be solutions of~$L_1,\dots,L_r$, respectively, 
  and consider
  \[
    w := \begin{vmatrix}
      f_1 & f_2 & \cdots & f_r \\
      \partial\cdot f_1 & \partial\cdot f_2 & \cdots & \partial\cdot f_r \\
      \vdots & \vdots & \ddots & \vdots \\
      \partial^{r-1}\cdot f_1 & \partial^{r-1}\cdot f_2 & \cdots & \partial^{r-1}\cdot f_r
      \end{vmatrix}.
  \]
  Then there exists an operator $M\in R[x][\partial]$ with $M\cdot w=0$ and
  \begin{alignat*}1
    \ord(M)&\leq r^r=:m,\qquad
     \deg(M)\leq m^2 r^2 d,\\
    \height(M)&\leq \height(m!) + m\height(1) 
      + (m-1)\height(mr^2d)\\
      &\kern-2em{} + m^2 r\bigl((r+1)(\height(4) + \height(r)) + \height(d) + c^{(m)}(d, h)\bigr).
  \end{alignat*}
\end{corollary}
\begin{proof}
  Apply Theorem~\ref{thm:main} with $n=r$, $P\in\R$ the polynomial obtained by
  replacing $f_i$ by $y_{i,0}$ in the expression given for~$w$. Note that
  $\Ord(P)<(r,\dots,r)$, $\Deg(P)=(1,\dots,1)$, $\deg(P)=0$, and
  $\height(P)=\height(1)$.
\end{proof}

\begin{experiment}
To check the bounds of Theorem~\ref{thm:main} for plausibility, we have computed
the symmetric product $L=L_1\otimes L_2$ for two random operators
$L_1,L_2\in\set Z[x][\partial]$ of order and degree and height bounded by~$s$,
for $s=2,3,4,5$. It turned out that the order of $L$ meets the bound stated in
the theorem. The bounds for degree and height are not as tight, but the data
suggests that they are only off by some constant factor. The results are given
in the table below.
\begin{center}
  \begin{tabular}{c|c|c|c|c}
    $s$ & 2 & 3 & 4 & 5 \\\hline
    degree bound & 64 & 486 & 2048 & 6250 \\
    actual degree & 16 & 90 & 320 & 850 \\ \hline
    height bound & 471.5 & 3495. & 14677. & 44980.2 \\ 
    actual height & 23.29 & 185.12 & 865.95 & 2693.30 
  \end{tabular}
\end{center}
\end{experiment}

\subsection{Order-Degree Curve}

Finally, the following result provides an order-degree curve for operators which
annihilates a function that is given as a polynomial of $f_1,\dots,f_n$ and
their derivatives/shifts. Once more, the technical difference in the argument
is that coefficient comparison is done with respect to the variables~$y_{i,j}$
as well as~$x$, giving a linear system over $R$ rather than over~$R[x]$. 

\begin{theorem}\label{thm:poly:curve}
  Let $L_1,\dots,L_n\in R[x][\partial]$, $r_i=\ord(L_i)$, $d_i=\deg(L_i)$.
  Let $\mathfrak{a}\subseteq\R$ be as in Lemma~\ref{lemma:4}. 
  Let $P\in\R$ be a homogeneous polynomial with 
  $\Ord(P)<(r_1,\dots,r_n)$ and $\Deg(P)=(D_1,\dots,D_n)$.
  Let
  \[
    r\geq m:=\prod_{i=1}^n\tbinom{D_i+r_i-1}{D_i}
    \text{ and }
    d\geq \frac{r\,m\sum\limits_{i=1}^n D_id_i + m\deg(P)}{r + 1 - m}.
  \]
  Then there exists an operator $L\in R[x][\partial]\setminus\{0\}$ 
  and a polynomial $p\in R[x]\setminus\{0\}$
  with $pL\cdot P\in\mathfrak{a}$ 
  and $\ord(L)\leq r$ and $\deg(L)\leq d$. 
\end{theorem}
\begin{proof}
  For $k=0,\dots,r$, let $V_k$ be as in the proof of Theorem~\ref{thm:main} but with $r$ in place of~$m$
  so that $\Ord(V_k)<(r_1,\dots,r_n)$ and $\deg(V_k)\leq\deg(P)+r\sum_{i=1}^n D_id_i$.   
  Make an ansatz $L = \sum_{i=0}^r\sum_{j=0}^d \ell_{i,j}x^j \partial^i$ for an operator of order~$r$
  and degree~$d$. We wish to determine the $\ell_{i,j}$ such that 
  \[
    \sum_{i=0}^r\sum_{j=0}^d \ell_{i,j} x^j V_i = 0.
  \]
  Coefficient comparison gives a linear system over $R$ with $(r+1)(d+1)$ variables and
  \[
    \max_{k=1}^n(d+1+\deg(V_k))m 
   =m\biggl(d+1+\deg(P)+r\sum_{i=1}^n D_id_i\biggr)
  \]
  equations. For $r$ and $d$ as in the theorem, there are more variables than equations,
  and therefore a nontrivial solution.
\end{proof}

\begin{experiment}
  From the algebra $\set Z[x][\partial]$ with $\sigma(x)=x+1$ and $\delta=0$ we
  picked three random operators $L_1,L_2,L_3$ of order, degree, and height~3,
  and we computed operators $L$ annihilating the Wronskian $w$ associated to
  these operators (cf.~Cor.~\ref{corr:wr} above). In the following figure we
  compare the degree bound obtained by last year's result~\cite{jaroschek13a}
  from the minimal order operator $L$ (dotted) to the a-priori degree bound of
  Theorem~\ref{thm:poly:curve} (solid). That the new bound overshoots is the
  price we have to pay for the feature that this bound can be calculated without
  knowing~$L$.

  \begin{center}
    \includegraphics[width=.8\hsize]{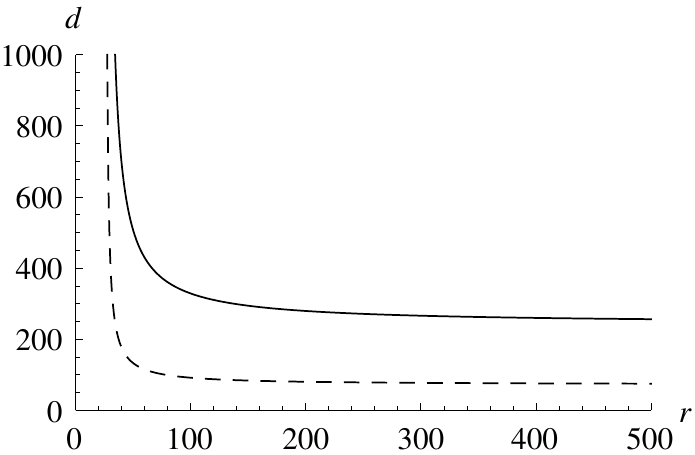}
  \end{center}
\end{experiment}




\bigskip
\noindent\textbf{Acknowledgement.} 
I would like to thank Ruyong Feng for asking a question I try to answer by this paper,
and the referees for pointing out some subtle mistakes in the first version. 

\bibliographystyle{plain}
\bibliography{all}

\end{document}